\documentclass[a4paper,10pt,oneside]{article}

\usepackage{amssymb}
\usepackage{amsmath}
\usepackage{amsthm}

\makeatletter
\let\@afterindentfalse\@afterindenttrue
\@afterindenttrue
\makeatother

\newtheorem{theorem}{Theorem}
\newtheorem{lemma}{Lemma}
\newtheorem{definition}{Definition}
\newtheorem{corollary}{Corollary}

\newcommand*{\ce}{\mathop{\mathrm{e}}\nolimits}
\newcommand*{\ci}{\mathop{\mathrm{i}}\nolimits}
\newcommand*{\id}{\mathop{\mathrm{id}}\nolimits}
\newcommand*{\Cov}{\mathop{\mathrm{Cov}}\nolimits}
\newcommand*{\qCov}{\mathop{\mathrm{qCov}}\nolimits}
\newcommand*{\gCov}{\mathop{\mathfrak{Cov}}\nolimits}
\newcommand*{\N}{\mathbb{N}}
\newcommand*{\R}{\mathbb{R}}
\newcommand*{\C}{\mathbb{C}}
\newcommand*{\Rp}{\mathbb{R}^{+}}
\newcommand*{\dint}{\mathop{\mathrm{d}}\nolimits}
\newcommand*{\Diag}{\mathop{\mathrm{Diag}}\nolimits}
\newcommand*{\Fop}{\mathcal{F}_{\mathrm{op}}}
\newcommand*{\cqre}[1]{\left[ {#1}_{0}\right]}
\newcommand*{\CM}{\mathcal{C_{M}}}
\newcommand*{\MN}{\mathcal{M}_{n}}
\newcommand*{\MNN}{\mathcal{M}_{n}^{1}} 
\newcommand*{\TMN}{M_{n,\mathrm{sa}}}
\newcommand*{\TMNN}{M_{n,\mathrm{sa}}^{(0)}}
\newcommand*{\Tr}{\mathop{\mathrm{Tr}}\nolimits}
\newcommand*{\Var}{\mathop{\mathrm{Var}}\nolimits}
\newcommand*{\scal}[2]{\left\langle #1, #2\right\rangle}
\newcommand*{\abs}[1]{\left\vert #1 \right\vert}
\newcommand*{\gz}[1]{\left( #1 \right)}
\newcommand*{\sz}[1]{\left\lbrack #1 \right\rbrack}
\newcommand*{\kz}[1]{\left\lbrace #1 \right\rbrace}

\begin{document}
\title{Refinement of Robertson-type uncertainty principles
with geometric interpretation
\thanks{keywords: uncertainty principle, quantum Fisher information;
          MSC: 62B10, 94A17}}

\author{Attila Lovas\thanks{lovas@math.bme.hu}, 
        Attila Andai\thanks{andaia@math.bme.hu}\\
  Department for Mathematical Analysis, \\
  Budapest University of Technology and Economics,\\
  H-1521 Budapest XI. Sztoczek u. 2, Hungary}
\date{May 14, 2015}

\maketitle

\begin{abstract}
A generalisation of the classical covariance for quantum mechanical observables 
  has previously been presented by Gibilisco, Hiai and Petz. Gibilisco and Isola 
  has proved that the usual quantum covariance gives the sharpest inequalities 
  for the determinants of covariance matrices.
We introduce a new generalisation of the classical covariance which gives better 
  inequalities than the classical one, furthermore it has a direct geometric 
  interpretation. 
\end{abstract}

\section*{Introduction}

The concept of uncertainty was introduced by Heisenberg in 1927 \cite{Hei}, who 
  demonstrated the impossibility of simultaneous measurement of position
  ($q$) and momentum ($p$).
He considered Gaussian distributions ($f(q)$), and defined uncertainty of $f$ as 
  its width $D_{f}$.
If the width of the Fourier transform of $f$ is denoted by $D_{\mathcal{F}(f)}$, 
  then the first formalisation of the uncertainty principle can be written as
\begin{equation*}
D_{f}D_{\mathcal{F}(f)}=\mbox{constant}.
\end{equation*}

\medskip
From now, the states of an $n$-level system are identified with 
  the set of $n\times n$ self adjoint positive semidefinite 
  matrices with trace $1$, and the physical observables are 
  identified with $n\times n$ self adjoint matrices.
   
In 1927 Kennard generalised Heisenberg's result \cite{Ken}, he proved the inequality
\begin{equation*}
\Var_{D}(A)\Var_{D}(B)\geq\frac{1}{4}
\end{equation*}
  for observables $A,B$ which satisfy the relation $\sz{A,B}=-\ci$, for
  every state $D$, where $\Var_{D}(A)=\Tr(DA^{2})-(\Tr(DA))^{2}$.
In 1929 Robertson \cite{Rob1} extended Kennard's result for arbitrary two 
  observables $A,B$
\begin{equation*}
\Var_{D}(A)\Var_{D}(B)\geq \frac{1}{4}\abs{\Tr(D\sz{A,B})}^2.
\end{equation*}
In 1930 Scr\"odinger \cite{Sch} improved this relation including the correlation 
  between observables $A,B$
\begin{equation*}
\Var_{D}(A)\Var_{D}(B)-\Cov_{D}(A,B)^{2}\geq 
  \frac{1}{4}\abs{\Tr(D\sz{A,B})}^2,
\end{equation*}
  where for a given state $D$ the (symmetrized) covariance of the observables 
  $A$ and $B$ is defined as
\begin{equation*}
\Cov_{D}(A,B)=\frac{1}{2}\gz{\Tr(DAB)+\Tr(DBA)}-\Tr(DA)\Tr(DB).
\end{equation*}
The Schr\"odinger uncertainty principle can be formulated as
\begin{equation*}
\det\begin{pmatrix}
\Cov_{D}(A,A) & \Cov_{D}(A,B)\\
\Cov_{D}(B,A) & \Cov_{D}(B,B)\end{pmatrix} \geq
\det\gz{-\frac{\ci}{2}\begin{pmatrix}
\Tr(D\sz{A,A}) & \Tr(D\sz{A,B}) \\
\Tr(D\sz{B,A}) & \Tr(D\sz{B,B})\end{pmatrix}}.
\end{equation*}
For the set of observables $(A_{i})_{1,\dots,N}$ this 
  inequality was generalised by Robertson in 1934 \cite{Rob2} as
\begin{equation*}
\det\gz{\sz{\Cov_{D}(A_h,A_j)}_{h,j=1,\dots,N}}
  \geq 
\det\gz{\sz{- \frac{\ci}{2}\Tr(D\sz{A_h,A_j})}_{h,j=1,\dots,N}}.
\end{equation*}
The main drawback of this inequality is that the right-hand side is identical 
  to zero whenever $N$ is odd.

Gibilisco and Isola in 2006 conjectured that 
\begin{equation}\label{eq:GibiliscoConjecture}
\det\gz{\sz{\Cov_{D}(A_h,A_j)}_{h,j=1,\dots,N}} \geq
\det\gz{
\sz{\frac{f(0)}{2}\scal{\ci\sz{D,A_{h}}}{\ci\sz{D,A_{j}}}_{D,f}}_{h,j=1,\dots,N}},
\end{equation}
  holds \cite{GibIso}, where the scalar product 
  $\scal{\cdot}{\cdot}_{D,f}$ is induced by an operator monotone
  function $f$, according to Petz classification theorem \cite{PetSud}.
We note that if the density matrix is not strictly positive, then the scalar 
  product $\scal{\cdot}{\cdot}_{D,f}$ is not defined.
The inequality \eqref{eq:GibiliscoConjecture} was studied first only in the 
  case $N=1$ for special functions $f$.
The cases $f(x)=f_{SLD}(x)=\frac{1+x}{2}$ and $f(x)=f_{WY}(x)=\frac{1}{4}(\sqrt{x}+1)^{2}$ were
  proved by Luo in \cite{Luo1} and \cite{Luo2}.
The general case of the conjecture was been proved by Hansen in \cite{Han} and 
  shortly after by Gibilisco et al. with a different technique in 
  \cite{GibImpIso1}.

In the case $N=2$ the inequality was proved for $f=f_{WY}$ by Luo, Q. Zhang and 
  Z. Zhang \cite{LuoZha1} \cite{LuoZha2} \cite{LuoZha3}.
The case of Wigner--Yanase--Dyson metric, where
  $f_{\beta}(x)=\frac{\beta(1-\beta)(x-1)^{2}}{(x^{\beta}-1)(x^{1-\beta}-1)}$
  ($\beta\in\left[ -1,2\right]\setminus\kz{0,1}$) was proved 
  independently by Kosaki \cite{Kos} and by Yanagi, Furuichi and Kuriyama 
  \cite{YanFurKur}.
The general case is due to Gibilisco, Imparato and Isola \cite{GibIso}   
  \cite{GibImpIso1}.

For arbitrary $N$ the conjecture was proved by Andai \cite{And} and 
  Gibilisco, Imparato and Isola \cite{GibImpIso3}.
The inequality \eqref{eq:GibiliscoConjecture} is called
  \emph{dynamical uncertainty principle} \cite{GibFumPetz}
  because the right-hand side can be interpreted as the volume of a 
  parallelepiped determined by the tangent vectors of the time-dependent 
  observables $A_k(t)=\ce^{\ci tD}A_k \ce^{-\ci tD}$.

Gibilisco, Hiai and Petz studied the behaviour of a possible generalization of 
  the covariance under coarse graining and they deduced that the covariance
  must have the following form for traceless observables $A,B$
\begin{equation}\label{petzcov}
\Cov_D^f (A,B) = \Tr \gz{A  f(L_{n,D}R_{n,D}^{-1})R_{n,D}(B)},
\end{equation}
  where $L_{n,D}$ and $R_{n,D}$ are superoperators acting on $n\times n$ 
  matrices like $L_{n,D}(A) = DA$, $R_{n,D}(A) = AD$ and $f$ is a symmetric and 
  normalized operator monotone function \cite{GibFumPetz}. 
Quantum covariances of the form \eqref{petzcov} are called
  \emph{quantum $f$-covariance}.
It has been proved \cite{GibFumPetz} that the generalized form of dynamical 
  uncertainty principle holds true for an arbitrary quantum $f$-covariance
\begin{equation*}
\det\gz{\sz{\Cov_{D}^g (A_h,A_j)}_{h,j=1,\dots,N}} \geq
\det\gz{
\sz{f(0) g(0)\scal{\ci\sz{D,A_{h}}}{\ci\sz{D,A_{j}}}_{D,f}}_{h,j=1,\dots,N}}
\end{equation*}
  and for all $g$ symmetric and normalized operator monotone function. 
If $g(x)=\frac{1+x}{2}$ is chosen, then we get the sharpest form of the
  inequality.

In this paper we will introduce the concept of the 
  \emph{symmetric $f$-covariance} as the scalar product of anticommutators 
\begin{equation*}
\qCov^{s}_{D,f}(A,B)=\frac{f(0)}{2}\scal{\kz{D,A}}{\kz{D,B}}_{D,f}
\end{equation*}
which has a clear geometric meaning.
Note that, $\qCov^{s}_{D,f}(A,B)$ and $\Cov_D^f(A,B)$ coincides with 
  $\Cov_D(A,B)$ whenever $f(x)=\frac{1+x}{2}$. 
We will prove that for any $f$ symmetric and normalized operator monotone function
\begin{equation*}
\det\gz{
\sz{\frac{f(0)}{2} \scal{\kz{D,A_{h}}}{\kz{D,A_{j}}}_{D,f}}_{h,j=1,\dots,N}}
\geq
\det\gz{\sz{\frac{f(0)}{2}   
  \scal{\ci\sz{D,A_{h}}}{\ci\sz{D,A_{j}}}_{D,f}}_{h,j=1,\dots,N}}
\end{equation*}
  holds which gives better inequalities than the dynamical uncertainty
  principle. 
Moreover we show that the function $f(x)=\frac{1}{2}\gz{\frac{1+x}{2}+\frac{2x}{1+x}}$ 
  gives the smallest universal upper bound for the right-hand side, that is, for
  every symmetric and normalized operator monotone function $g$
\begin{equation*}
\det\gz{
\sz{\frac{f(0)}{2} \scal{\kz{D,A_{h}}}{\kz{D,A_{j}}}_{D,f}}_{h,j=1,\dots,N}}
\geq
\det\gz{\sz{\frac{g(0)}{2}   
  \scal{\ci\sz{D,A_{h}}}{\ci\sz{D,A_{j}}}_{D,g}}_{h,j=1,\dots,N}}
\end{equation*}
  holds.

To make the paper self contained and understandable by the largest possible 
  audience, in Section 1, we briefly outline the geometry of the state space and 
  we give a simple proof for the general form of uncertainty relation which was 
  presented in \cite{GibFumPetz} and which will be used in the sequel.
In Section 2, we define the symmetric and antisymmetric $f$-covariances of
  observables, depending on an operator monotone function and we give  
  their local form.
The main inequalities can be found in Section 3, where we prove that the symmetric
  covariance gives better upper bound for the dynamical covariance and we
  estimate the sharpness of the new upper bound.
In Section 4 we present the smallest universal upper bound 
  for dynamical uncertainty principles determined by any operator monotone function.

\section{Riemannian metrics on the state space}

Let us denote by $\MN$ the set of $n\times n$ positive definite
  matrices and by $\MNN$ the interior of the $n$-level quantum mechanical state
  space i.e. the set of $n\times n$ positive definite trace one  
  matrices, namely $\MNN = \kz{D\in\MN |\Tr D = 1,\,D>0}$.
Let $\TMN$ be the set of observables of the $n$-level quantum system,
  in other words the set of $n\times n$ self adjoint matrices, and
  let $\TMNN$ be the set observables with zero trace.
Spaces $\MN$ and $\MNN$ are form convex sets in the space of self adjoint 
  matrices, and they are obviously differentiable manifolds \cite{HiaPetTot}.
The tangent space of $\MN$ at a given state $D$  can be identified 
  with $\TMN$ and the tangent space of $\MNN$ with $\TMNN$.

\begin{definition}  
Let $(K^{(m)})_{m\in\N}$ be  family of Riemannian metrics on $\MNN$.
This family of metrics is said to be {\em monotone} if
\begin{equation*}
K_{T(D)}^{(m)}(T(X),T(X))\leq K_{D}^{(n)}(X,X)
\end{equation*}
  holds for every completely positive, trace preserving linear map
  $T:M_{n}(\C)\to M_{m}(\C)$ (such a mapping is called a stochastic mapping), 
  for every $D\in\MNN$ and for all $X\in\TMNN$ and $m,n\in\N$.
\end{definition}

Let us denote by $\Fop$ the set of operator monotone functions 
  $f:\Rp\to\R$ with the property $f(x)=xf(x^{-1})$ for every 
  $x\in\Rp$ and with the normalization condition $f(1)=1$.
These functions are called symmetric and normalized operator
  monotone functions and they play a crucial role in theory of
  monotone metrics.
For the mean induced by the operator monotone function
  $f\in\Fop$ we also introduce the notation
\begin{equation*}
m_{f}:\Rp\times\Rp\to\Rp\qquad (x,y)\mapsto yf\gz{\frac{x}{y}}.
\end{equation*} 
The reciprocal of $m_f$ is called the Chentsov--Morozova function
  associated to the function $f\in\Fop$.
Examples of elements of $\Fop$ are presented in the following list
\begin{center}
\begin{tabular}{lll}
$f_{RLD}(x)=\frac{2x}{1+x}$, 
  & $f_{SLD}(x)=\frac{1+x}{2}$, 
  &  $f_{P1}(x)=\frac{2x^{\alpha+\frac{1}{2}}}{1+x^{2\alpha}}$ 
    $\quad\alpha\in\sz{0,\frac{1}{2}}$,\\
$f_{WY}(x)=\gz{\frac{1+\sqrt{x}}{2}}^2$, 
  & $f_{KM}(x)=\frac{x-1}{\log x}$, 
  & $f_{WYD}(x)=\beta (1-\beta)\frac{(x-1)^2}{(x^\beta-1)(x^{1-\beta}-1)}$ 
  $\quad\beta\in\left]0,\frac{1}{2}\right[$.
\end{tabular}
\end{center}
Let us introduce the sets of regular and non-regular
  elements in $\Fop$
\begin{equation*}
\Fop^r=\kz{f\in\Fop|f(0)\ne 0},  \quad \Fop^n=\kz{f\in\Fop|f(0) = 0}
\end{equation*}
  for which trivially $\Fop^r \cup\Fop^n=\Fop$ and $\Fop^r \cap\Fop^n=\emptyset$ 
  holds.
Next Theorem establishes a bijection between $\Fop^r$ and $\Fop^n$ \cite{GibHanIso}. 

\begin{theorem}\label{thm:FopRN}
If we define for every function $f$ in $\Fop^r$
\begin{equation*}
\tilde{f}(x)=\frac{1}{2} \sz{(1+x)-(1-x)^2\frac{f(0)}{f(x)}},
\end{equation*}
 then the correspondence $f\mapsto\tilde{f}$ is a bijection between 
 $\Fop^r$ and $\Fop^n$.
\end{theorem}
The mapping $\frac{f(0)}{f(x)}\mapsto\tilde{f}(x)$
  reverses the order which implies $\frac{f(0)}{f(x)}$ reaches its maximum when 
  $f=f_{SLD}$.

\begin{theorem}
{\it Petz classification theorem \cite{PetSud}.}
\label{th:Petz}
There exists a bijective correspondence between the monotone family of metrics
  $(K^{(n)})_{n\in\N}$ and functions $f\in\Fop$.
The metric is given by
\begin{equation}
\label{eq:generatedscalprod}
K_{D}^{(n)}(X,Y)=\Tr\gz{X m_f\gz{L_{n,D},R_{n,D}}^{-1}(Y)}
\end{equation}
  for all $n\in\N$ where $L_{n,D}(X)=DX$, $R_{n,D}(X)=XD$ for all
  $D,X\in M_{n}(\C)$.
\end{theorem}
 
The metric defined in Theorem \ref{th:Petz} can be extended to the space $\MN$. 
For every $D\in\MN$ and matrices $A,B\in\TMN$ let us define
\begin{equation*}
\scal{A}{B}_{D,f}=\Tr\gz{A  m_f\gz{L_{n,D},R_{n,D}}^{-1}(B)},
\end{equation*}
  with this notion the pair $(\MN,\scal{\cdot}{\cdot}_{\cdot,f})$ will be a
  Riemannian manifold for every operator monotone function $f\in\Fop$.

Next, we introduce Riemannian metrics in more general form 
  on the state space for which determinant inequalities leading to
  uncertainty principles will be proved.

\begin{definition}
Let us introduce the notation 
\begin{equation*}
\CM=\kz{g:\Rp\times\Rp\to\Rp\Bigm\vert\ 
\begin{array}{l}
  \mbox{$g$ is a symmetric smooth function, with analytical}\\
  \mbox{extension defined on a neighborhood of $\Rp\times\Rp$}
\end{array}}
\end{equation*}
  implicating the correspondence with Chentsov--Morozova functions.
Fix a function $g\in\CM$.
Define for every $D\in\MN$ and for every $A,B\in\TMN$ 
\begin{equation*}
(A,B)_{D,g}=\Tr\gz{A g(L_{n,D},R_{n,D})(B)}.
\end{equation*}
\end{definition}

\begin{theorem}
\label{th:g-metric}
For every $g\in\CM$ the function $(\cdot,\cdot)_{\cdot,g}$
  defined above is a Riemannian metric on $\MN$.
If the state $D\in\MNN$ is of the form $D=\Diag(\lambda_{1},\dots,\lambda_{n})$, 
  then for every $D\in\MN$, $A,B\in\TMN$ and $g\in\CM$ the
  local form of $(A,B)_{D,g}$ can be expressed as follows
\begin{equation*}
(A,B)_{D,g}=\sum_{k,l=1}^{n}A_{lk}B_{kl}g(\lambda_{k},\lambda_{l}).
\end{equation*}
\end{theorem}
\begin{proof}
Note, that the operators $L_{n,D}$ and $R_{n,D}$ are self adjoint and positive 
  definite with respect to the Hermitian form $(X,Y)\mapsto\Tr\gz{X^{*}Y}$.
For an analytical function $f$ the operators $f(L_{n,D})$ and $f(R_{n,D})$
  are well defined, and by the Riesz--Dunford operator calculus we have
\begin{equation*}
f(L_{n,D})=\frac{1}{2\pi\ci}\oint f(\xi)(\xi\id_{\TMN}-L_{n,D})^{-1}\dint \xi,
\end{equation*}
  where we integrate once around the spectrum of $D$.
The operators $L_{n,D}$ and $R_{n,D}$ commute, therefore a similar reasoning 
  gives the formula
\begin{equation*}
g(L_{n,D},R_{n,D})
=\frac{1}{(2\pi\ci)^{2}}\oint\oint g(\xi,\eta)(\xi\id-L_{n,D})^{-1}
  \circ (\eta\id-R_{n,D})^{-1}\dint\xi\dint\eta
\end{equation*}
  or equivalently
\begin{equation*}
g(L_{n,D},R_{n,D})(B)
=\frac{1}{(2\pi\ci)^{2}}\oint\oint g(\xi,\eta)(\xi-D)^{-1}
  B (\eta-D)^{-1}\dint\xi\dint\eta.
\end{equation*}
Assuming the matrix $D$ to be the form of
  $D=\Diag(\lambda_{1},\dots,\lambda_{n})$ for matrix units $E_{ij}$ and 
  $E_{kl}$ we have
\begin{equation*}
(E_{ij},E_{kl})_{D,g}
=\Tr\frac{1}{(2\pi\ci)^{2}}\oint\oint g(\xi,\eta)E_{ij}(\xi-D)^{-1}
  E_{kl}(\eta-D)^{-1}\dint\xi\dint\eta
=\delta_{kj}\delta_{li}g(\lambda_{j},\lambda_{i}).
\end{equation*}
Since for arbitrary matrices $A,B\in\TMN$ we have 
  $A=\sum_{i,j=1}^{n}A_{ij}E_{ij}$ and $B=\sum_{k,l=1}^{n}B_{kl}E_{kl}$, 
  therefore
\begin{equation*}
(A,B)_{D,g}
=\sum_{i,j,k,l=1}^{n}\delta_{kj}\delta_{li}A_{ij}B_{kl}g(\lambda_{j},\lambda_{i})
=\sum_{k,l=1}^{n}A_{lk}B_{kl}g(\lambda_{k},\lambda_{l}).
\end{equation*}
For every $g\in\CM$ and $D\in\MN$ the function $(A,B)\mapsto (A,B)_{D,g}$
  is a positive bilinear map, and for every smooth vector field 
  $\gamma:\MN\to\TMN$ the function
\begin{equation*}
\MN\to\R\qquad D\mapsto (\gamma(D),\gamma(D))_{D,g}
\end{equation*}
  is smooth, which means that $(\cdot,\cdot)_{\cdot,g}$ defines a 
  Riemannian metric on $\MN$.
\end{proof}
 
Following theorem states that Riemannian metrics corresponding to operator 
  monotone functions $f\in\Fop$ are special cases of metrics $(.,.)_{.,g}$ 
  for some $g\in\CM$.

\begin{theorem}
\label{th:classical-scalprod}
For every $D\in\MN$, $A,B\in\TMN$ and $f\in\Fop$ we have
\begin{equation*}
\scal{A}{B}_{D,f}
  =\sum_{k,l=1}^{n}A_{lk}B_{kl}\frac{1}{m_{f}(\lambda_{k},\lambda_{l})}.
\end{equation*}
\end{theorem}
\begin{proof}
Assume that $f\in\Fop$, $D\in\MN$ and $A,B\in\TMN$.
If we define 
\begin{equation*}
g:\Rp\times\Rp\to\Rp\qquad
  (x,y)\mapsto \frac{1}{yf\gz{\frac{x}{y}}},
\end{equation*}
  then we have $g(x,y)=\frac{1}{m_{f}(x,y)}$ and
\begin{equation*}
\scal{A}{B}_{D,f}=(A,B)_{D,g}.
\end{equation*}
Applying Theorem \ref{th:g-metric} we have
\begin{equation*}
(A,B)_{D,g}=\sum_{k,l=1}^{n}A_{lk}B_{kl}g(\lambda_{k},\lambda_{l})
  =\sum_{k,l=1}^{n}A_{lk}B_{kl}\frac{1}{m_{f}(\lambda_{k},\lambda_{l})},
\end{equation*}
  which proves Theorem \ref{th:classical-scalprod}.
\end{proof}

For observable $A\in\TMN$ and a state $D\in\MNN$ we define 
  $A_{0}=A-\Tr(DA)I$, where $I$ is the $n\times n$ identity matrix. 
Using this transformation we have $\Tr DA_{0}=0$.

The next theorem states that inequalities between Chentsov--Morozova functions 
  induce inequalities between the determinants of Gram matrices 
  corresponding to the considered Riemannian metrics.
Before the theorem we note that the idea of the proof will be referred
  many times later and similar reasoning will give us the desired uncertainty 
  inequalities.

\begin{theorem}\label{th:mainineq}
Consider a density matrix $D\in\MNN$, functions $g_{1},g_{2}\in\CM$ such that
\begin{equation*}
g_{1}(x,y) \geq g_{2}(x,y) \qquad \forall x,y\in\Rp
\end{equation*}
  holds and an $N$-tuple of nonzero matrices $(A^{(k)})_{k=1,\dots,N}\in\TMN$.
Define the $N\times N$ matrices $\gCov_{D,g_{1}}$ and $\gCov_{D,g_{2}}$ with 
  entries
\begin{equation*}
\sz{\gCov_{D,g_{k}}}_{ij}=(A^{(i)}_{0},A^{(j)}_{0})_{D,g_{k}}\qquad k=1,2.
\end{equation*}
We have
\begin{equation}
\label{eq:mainineq}
\det(\gCov_{D,g_{1}})\geq \det(\gCov_{D,g_{2}})+\det(\gCov_{D,g_{1}}-\gCov_{D,g_{2}})+R(D,g_{1},g_{2},N),        \end{equation}
  where
\begin{equation*}
R(D,g_{1},g_{2},N)=\sum_{k=1}^{N-1}\binom{N}{k}
 \sz{\det(\gCov_{D,g_{1}})}^{\frac{k}{N}}
 \sz{\det(\gCov_{D,g_{1}}-\gCov_{D,g_{2}})}^{\frac{N-k}{N}}.
\end{equation*}
\end{theorem}
\begin{proof}
We can assume that the density matrix $D$ is diagonal, 
  $D=\Diag(\lambda_{1},\dots,\lambda_{n})$.
The matrix $\gCov_{D,g_{k}}$ ($k=1,2$) is obviously real and symmetric.
First we prove that for every $g\in\CM$ the matrix $\gCov_{D,g}$ 
  is positive.
Consider a vector $x\in\mathbb{C}^{N}$ and define an $n\times n$ matrix as
  $C=\sum_{a=1}^{N}x_{a}A_{0}^{(a)}$.
Then we have
\begin{align}
\label{eq:posdef}
\scal{x}{\gCov_{D,g}x}
&=\sum_{a,b=1}^{N}\overline{x_{a}}x_{b}\gCov_{D,f}(A^{(a)},A^{(b)})\\
&=\sum_{a,b=1}^{N}\sum_{h,j=1}^{n}g(\lambda_{h},\lambda_{j})
  \overline{x_{a}}x_{b}\cqre{A^{(a)}}_{hj}\cqre{A^{(b)}}_{jh}\nonumber\\
&=\sum_{h,j=1}^{n}g(\lambda_{h},\lambda_{j})\overline{\left(\sum_{a=1}^{N}x_{a}\cqre{A^{(a)}}_{jh}\right)}
  \left(\sum_{b=1}^{N}x_{b}\cqre{A^{(b)}}_{jh}\right)\nonumber\\
&=\sum_{h,j=1}^{n}g(\lambda_{h},\lambda_{j})\vert C_{hj}\vert^{2}
  \geq 0.\nonumber
\end{align}

We can repeat our arguments providing Equation (\ref{eq:posdef}) for the matrix 
  $\gCov_{D,g_{1}}-\gCov_{D,g_{2}}$ using $g_{1}-g_{2}$ instead of $g$.
This leads us to the conclusion that $\gCov_{D,g_{1}}-\gCov_{D,g_{2}}$ is a
  real, symmetric, positive matrix.

Using the Minkowski determinant inequality 
  (see for example \cite{BecBel} p. 70. or \cite{HorJoh})
  for real symmetric positive matrices $\gCov_{D,g_{1}}$ and 
  $(\gCov_{D,g_{1}}-\gCov_{D,g_{2}})$ we have
\begin{equation*}
\kz{\det( \gCov_{D,g_{2}}+(\gCov_{D,g_{1}}-\gCov_{D,g_{2}}))}^{\frac{1}{N}}\geq
\kz{\det( \gCov_{D,g_{2}})}^{\frac{1}{N}}+
\kz{\det( \gCov_{D,g_{1}}-\gCov_{D,g_{2}})}^{\frac{1}{N}},
\end{equation*}
  which gives the inequality stated in the Theorem.
\end{proof}

\section{Symmetric $f$-covariances}

Now we define some well-known covariances and introduce the notion 
  of symmetric and antisymmetric covariances.

\begin{definition}
For observables $A,B\in\TMN$, state $D\in\MNN$ and function 
  $f\in\Fop$ we define 
  the {\em covariance} of $A$ and $B$
\begin{equation*}
\Cov_{D}(A,B)=\frac{1}{2}\gz{\Tr(DAB)+\Tr(DBA)}-\Tr(DA)\Tr(DB),
\end{equation*}
  the {\em quantum $f$-covariance}, which  was
  introduced by Gibilisco and Isola \cite{GibIsoNew} 
\begin{equation*}
\Cov_D^f (A,B) = \Tr \gz{A  f(L_{n,D}R_{n,D}^{-1})R_{n,D}(B)},
\end{equation*}
  the {\em antisymmetric $f$-covariance}  
\begin{equation*}
\qCov^{as}_{D,f}(A,B)=\frac{f(0)}{2}\scal{\ci\sz{D,A}}{\ci\sz{D,B}}_{D,f}
\end{equation*}
  and the {\em symmetric $f$-covariance} as
\begin{equation*}
\qCov^{s}_{D,f}(A,B)=\frac{f(0)}{2}\scal{\kz{D,A}}{\kz{D,B}}_{D,f},
\end{equation*}
where $[.,.]$ is the commutator of matrices
and $\kz{.,.}$ denotes the anticommutator respectively.
\end{definition}

Note that the symmetric and antisymmetric covariances are trivial 
  if $f\in\Fop^{n}$ (that is $f(0)=0$).
The local form of the covariances at a given state is the
  following.

\begin{theorem}
\label{th:covariances}
Assume that the state $D\in\MNN$ is of the form 
  $D=\Diag(\lambda_{1},\dots,\lambda_{n})$.
For every $A,B\in\TMN$  and $f\in\Fop$ we have
\begin{align*}
\Cov_{D}(A,B)
  &=\sum_{k,l=1}^{n}
  \frac{\lambda_{k}+\lambda_{l}}{2} A_{lk}B_{kl}-\Tr(DA)\Tr(DB)
  =\sum_{k,l=1}^{n}
  \frac{\lambda_{k}+\lambda_{l}}{2} \sz{A_{0}}_{lk}\sz{B_{0}}_{kl}\\
\Cov_D^f (A,B)
  &=\sum_{k,l=1}^{n}
  m_{f}(\lambda_{k},\lambda_{l}) A_{lk}B_{kl}\\
\qCov^{as}_{D,f}(A,B)
  &=\frac{f(0)}{2}\sum_{k,l=1}^{n}
  \frac{(\lambda_{k}-\lambda_{l})^{2}}{m_{f}(\lambda_{k},\lambda_{l})}
  A_{lk}B_{kl}
  =\frac{f(0)}{2}\sum_{k,l=1}^{n}
  \frac{(\lambda_{k}-\lambda_{l})^{2}}{m_{f}(\lambda_{k},\lambda_{l})}
  \sz{A_{0}}_{lk}\sz{B_{0}}_{kl}\\
\qCov^{s}_{D,f}(A,B)
  &=\frac{f(0)}{2}\sum_{k,l=1}^{n}
  \frac{(\lambda_{k}+\lambda_{l})^{2}}{m_{f}(\lambda_{k},\lambda_{l})}
  A_{lk}B_{kl}\\
\qCov^{s}_{D,f}(A,B)
  &=\frac{f(0)}{2}\sum_{k,l=1}^{n}
  \frac{(\lambda_{k}+\lambda_{l})^{2}}{m_{f}(\lambda_{k},\lambda_{l})}
  \sz{A_{0}}_{lk}\sz{B_{0}}_{kl}+2f(0)\Tr(DA)\Tr(DB).
\end{align*}
\end{theorem}
\begin{proof}
Simple matrix computation, we show the case of antisymmetric quantum covariance.
\begin{align*}
\qCov^{as}_{D,f}(A,B)
&=\frac{f(0)}{2}\scal{DA-AD}{DB-BD}_{D,f}\\
&=\frac{f(0)}{2}\sum_{k,l=1}^{n}\frac{1}{m_{f}(\lambda_{k},\lambda_{l})}
  (DA-AD)_{lk}(DB-BD)_{kl}
\end{align*}
Since $(DA-AD)_{lk}=(\lambda_{l}-\lambda_{k})A_{lk}$ we immediately have the 
  formula in the Theorem.

All of the covariances in the Theorem have the form of
\begin{equation*}
\sum_{k,l=1}^{n}\alpha_{kl}A_{kl}B_{kl},
\end{equation*}
  where $\alpha$ is a symmetric function.
Since $\sz{A_{0}}_{kl}=A_{kl}-\delta_{kl}\Tr(DA)$ we have the equality
\begin{align*}
\sum_{k,l=1}^{n}\alpha_{kl}\sz{A_{0}}_{kl}\sz{B_{0}}_{kl}
  =&\sum_{k,l=1}^{n}\alpha_{kl}A_{kl}B_{kl}
  +\gz{\sum_{k=1}^{n}\alpha_{kk}}\Tr(DA)\Tr(DB)\\
  &-\gz{\sum_{k=1}^{n}\alpha_{kk}A_{kk}}\Tr(DB)
  -\gz{\sum_{k=1}^{n}\alpha_{kk}B_{kk}}\Tr(DA).
\end{align*}
Using this equation one can express $\sum_{k,l=1}^{n}\alpha_{kl}A_{kl}B_{kl}$
  in terms of normalized observables.
\end{proof}

It is worth noting that for operator monotone function 
  $f_{SLD}(x)=\frac{1+x}{2}$ we have
\begin{equation*}
\Cov_{D}(A_{0},B_{0})=\qCov^{s}_{D,f_{SLD}}(A_{0},B_{0})
  =\Cov_{D}^f(A_{0},B_{0}).
\end{equation*}

Next we define the Chentsov--Morozova functions associated
  to $\qCov^{as}_{D,f}(A,B)$ and $\qCov^{s}_{D,f}(A,B)$.

\begin{definition}
For an operator monotone function $f\in\Fop$ let us define the following 
  $\Rp\to\Rp$ functions
\begin{equation*}
f_{as}(x)=\frac{f(0)(1-x)^2}{2f(x)},\quad
  f_s(x)=\frac{f(0)(1+x)^2}{2f(x)}
\end{equation*}
  and $\Rp\times\Rp\to\Rp$ functions
\begin{equation*}
g_{cl}(x,y)=\frac{x+y}{2},\quad
g^{as}_{f}(x,y)=\frac{f(0)(x-y)^{2}}{2\,m_f(x,y)},\quad
g^{s}_{f}(x,y)=\frac{f(0)(x+y)^{2}}{2\,m_f(x,y)}.
\end{equation*}
\end{definition}

With this notation we have 
\begin{equation*}
 g_f^{as}(x,y)=x f_{as}\gz{\frac{y}{x}}
  \qquad\mbox{and}\qquad 
  g_f^{s}(x,y)=x f_{s}\gz{\frac{y}{x}}.
\end{equation*}

A simple combination of Theorem \ref{th:covariances} and \ref{th:g-metric}
  gives the following Corollary.
\begin{corollary}
For every $f\in\Fop$, $D\in\MNN$ and for every matrices $A,B\in\TMN$ the
  following equalities hold.
\begin{align*}
(A_{0},B_{0})_{D,g_{cl}}&=\Cov_{D}(A_{0},B_{0})\\
(A_{0},B_{0})_{D,g^{as}_{f}}&=\qCov_{D,f}^{as}(A_{0},B_{0})
  =\qCov_{D,f}^{as}(A,B)\\
(A_{0},B_{0})_{D,g^{s}_{f}}&=\qCov_{D,f}^{s}(A_{0},B_{0})
\end{align*}
\end{corollary}

\section{Uncertainty principles}

For a fixed density matrix $D\in\MNN$, function $f\in\Fop$ 
  and an $N$-tuple of nonzero matrices $(A^{(k)})_{k=1,\dots,N}\in\TMN$
  we define the $N\times N$ matrices $\Cov_{D}$, 
  $\Cov_{D}^f$, $\qCov^{as}_{D,f}$ and $\qCov^{s}_{D,f}$ with entries
\begin{align*}
\sz{\Cov_{D}}_{ij}&=\Cov_{D}(A^{(i)}_{0},A^{(j)}_{0})\\
\sz{\Cov_{D}^f}_{ij}&=\Cov_{D}^f(A^{(i)}_{0},A^{(j)}_{0})\\
\sz{\qCov^{as}_{D,f}}_{ij}&
  =\qCov^{as}_{D,f}(A^{(i)}_{0},A^{(j)}_{0})\\
\sz{\qCov^{s}_{D,f}}_{ij}&
  =\qCov^{s}_{D,f}(A^{(i)}_{0},A^{(j)}_{0}).
\end{align*}
Using the inequality \eqref{eq:mainineq} in Theorem \ref{th:mainineq}
  one can compare the different covariances.
  
\begin{theorem}\label{thm:unc}
If $f_{1},f_{2}$ are functions symmetric and normalized operator
  monotone functions for which
\begin{equation*}
\frac{f_{1}(0)}{f_{1}(t)}\geq \frac{f_{2}(0)}{f_{2}(t)}
 \qquad \forall t\in\Rp
\end{equation*}
  holds, then we have
\begin{align*}
\det(\Cov_{D})
  &\geq \det(\qCov^{s}_{D,f_{k}})
  \geq \det(\qCov^{as}_{D,f_{k}})  \qquad k=1,2\\
\det(\qCov^{as}_{D,f_{1}})&\geq \det(\qCov^{as}_{D,f_{2}})\\
\det(\qCov^{s}_{D,f_{1}})&\geq \det(\qCov^{as}_{D,f_{2}}).
\end{align*}
\end{theorem}
\begin{proof}
If for functions $f_{1},f_{2}\in\Fop$ 
\begin{equation*}
\frac{f_{1}(0)}{f_{1}(t)}\geq \frac{f_{2}(0)}{f_{2}(t)}
 \qquad \forall t\in\Rp
\end{equation*}
  holds, then we have the following inequalities for all $x,y\in\Rp$.
\begin{align*}
g_{cl}(x,y)
  &\geq g_{f_{k}}^{s}(x,y)\geq g_{f_{k}}^{as}(x,y)  \qquad k=1,2\\
g_{f_{1}}^{as}(x,y)&\geq g_{f_{2}}^{as}(x,y)\\
g_{f_{1}}^{s}(x,y)&\geq g_{f_{2}}^{s}(x,y)  
\end{align*}
Applying Theorem \ref{th:mainineq} for these inequalities we have the
  inequalities for different covariance matrices stated in the Theorem.
\end{proof}

\begin{corollary}
Using the same notation as in the previous Theorem for any operator monotone
  function $f\in\Fop$ we have
\begin{equation*}
\det(\Cov_{D})\geq \det(\qCov^{s}_{D,f})\geq \det(\qCov^{as}_{D,f}).
\end{equation*}
\end{corollary}

The dynamical uncertainty relation mentioned in the Introduction 
  (see Equation \eqref{eq:GibiliscoConjecture}) in these terms is
\begin{equation*}
\det(\Cov_{D})\geq\det(\qCov^{as}_{D,f}).
\end{equation*}

It is surprising that for a given $f\in\Fop$ the gap between
  symmetric and antisymmetric \mbox{$f$-covariances}
  is so large that it can be estimated below by
  the quantum $f_{RLD}$-covariance.

\begin{theorem}
For a density matrix $D\in\MNN$ and function $f\in\Fop$ we have
 \begin{equation*}
  2f(0) \Cov_D^{f_{RLD}}(A_0,B_0)
  \le
  \qCov_{D,f}^s(A_0,B_0)-\qCov_{D,f}^{as}(A_0,B_0)
  \le
  \Cov_D^{f_{RLD}}(A_0,B_0)
 \end{equation*}
  and
 \begin{equation*}
  \det(\qCov^{s}_{D,f})-\det(\qCov^{as}_{D,f})
  \ge
  (2f(0))^N  \det(\Cov_D^{f_{RLD}}).
 \end{equation*}
\end{theorem}
\begin{proof}
Consider the involution 
\begin{equation*}
 \#:\Fop\to\Fop\qquad f(x)\mapsto f^\#(x)=\frac{x}{f(x)}
\end{equation*}
  cf. Definition 2.5 in \cite{Aud}. 
For any function $g\in\Fop$ one has $f_{RLD}(x)=\frac{2x}{1+x}\le g(x)$ 
  thus we have
\begin{equation*}
f_s(x)-f_{as}(x)=\frac{2f(0)x}{f(x)}=2f(0)f^\#(x)\ge 2f(0) f_{RLD}(x)
\end{equation*}
  which implies the lower bound in the first inequality.
According to the remark after Theorem \ref{thm:FopRN}
\begin{equation*}
\frac{2f(0)x}{f(x)}\le\frac{2x}{1+x}
\end{equation*}
  which gives the upper bound in the first inequality.

By using Inequality \eqref{eq:mainineq} twice, 
  once
  for $g_1(x,y) = g_f^{s}(x,y)$ 
  and $g_2(x,y) = g_f^{as}(x,y)$
  and 
  once for $g_1(x,y) = g_f^{s}(x,y)-g_f^{as}(x,y)$
  and $g_2(x,y) = 2f(0)m_{f_{RLD}}(x,y)$
  we get
 \begin{equation*}
  \det(\qCov^{s}_{D,f})-\det(\qCov^{as}_{D,f})
  \ge
  \det(\qCov^{s}_{D,f}-\qCov^{as}_{D,f})
  \ge
  (2f(0))^N  \det(\Cov_D^{f_{RLD}}).
 \end{equation*}
which is the desired second inequality.
\end{proof}

\section{The global upper bound}

It is natural to ask, if there exists an operator monotone 
  function $f\in\Fop$ for which
\begin{equation*}
\det (\qCov_{D,f}^s)\ge \det (\qCov_{D,g}^{as})
\end{equation*}
  holds for any $g\in\Fop$. 
In this Section we show the existence of such $f$ function and we 
  characterise these functions and we present the optimal one.

According to Theorem \ref{thm:unc} it is enough to find a
  function $f\in\Fop^r$ for which
\begin{equation}
\label{eq:globupperboundineq}
 \frac{f(0)(1+x)^2}{2f(x)}
 \ge
 \frac{(1-x)^2}{2(1+x)}
\end{equation}
  holds. 
This condition can be reformulated by using the remark after 
  Theorem \ref{thm:FopRN} as
\begin{equation}\label{eq:cond}
 \tilde{f}(x)
 \le
 \frac{2x}{1+x}
 \sz{1+\gz{\frac{1-x}{1+x}}^2}.
 \end{equation}

For proving the existence of such function we need some
  more sophisticated tool. 
Next theorem establishes a canonical bijection between 
  elements in $\Fop$ and probability measures on $\sz{0,1}$ 
  (see \cite{Aud,Han}).

\begin{theorem}
For a given $g\in\Fop$ there exists a unique probability
  measure on $\sz{0,1}$ such that for all $x\in\Rp$
 \begin{equation*}
 \frac{1}{g(x)}=
 \int\limits_{\sz{0,1}}
 \frac{1+t}{2} \gz{\frac{1}{x+t} + \frac{1}{1+tx}}
 \dint\mu(t).
 \end{equation*}
\end{theorem}

\begin{lemma}
If $\mu$ denotes the probability measure on $\sz{0,1}$ associated 
  to the function $g\in\Fop$ and $\mu (\kz{0})<\frac{1}{2}$, 
  then there exists $x>0$ such that
 \begin{equation*}
  \frac{1}{g(x)}
  <\frac{1+x}{2x}\frac{1}{1+\gz{\frac{1-x}{1+x}}^2}
 \end{equation*}
  holds.
\end{lemma}
\begin{proof}
Choose $\varepsilon > 0$ so small that 
  $\mu(\sz{0,\varepsilon})<\frac{1}{2}$ and consider the 
  following estimation
\begin{equation*}
\frac{1}{g(x)}=
 \int\limits_{\sz{0,1}}
 \frac{1+t}{2}\gz{\frac{1}{x+t}+\frac{1}{1+tx}} \dint\mu(t)
 \le
 \frac{1+x}{2x}\mu(\sz{0,\varepsilon})
 +\mu(\left]\varepsilon,1\right])\frac{(1+\varepsilon)^2}{2\varepsilon}.
\end{equation*}
Since the function
\begin{equation*}
\Rp\to\Rp\qquad
  x\mapsto\frac{1+x}{2x}\frac{1}{1+\gz{\frac{1-x}{1+x}}^2}
\end{equation*} 
  is bounded below by $\frac{1}{2}$, which implies
\begin{equation*}
 \frac{1+x}{2x}\mu(\sz{0,\varepsilon})
 +\mu(\left]\varepsilon,1\right])\frac{(1+\varepsilon)^2}{2\varepsilon}
 <\frac{1+x}{2x}\frac{1}{1+\gz{\frac{1-x}{1+x}}^2}
\end{equation*}
for $x>0$ small enough.
\end{proof}

Let $\mu$ be a probability measure on $\sz{0,1}$ defined by
  $\mu(\kz{0})=\mu(\kz{1})=\frac{1}{2}$.
For the function $g\in\Fop$ corresponding to $\mu$ we have 
\begin{equation*}
\frac{1}{g(x)}=\frac{1}{2}\gz{\frac{x+1}{2x}+\frac{2}{x+1}}.
\end{equation*}
For every $x\in\Rp$ we have
\begin{equation*}
\frac{2x}{1+x}\sz{1+\gz{\frac{1-x}{1+x}}^2}-g(x)
  =\frac{(1-x)^2}{(1+x)(1+x^2)}>0,
\end{equation*}
  which means that $g(x)$ is a good candidate for 
  $\tilde{f}$ in \eqref{eq:cond} and by the inversion 
  formula (See Proposition 6.1. of the paper \cite{GibHanIso}) we have that
\begin{equation*}
f(x) =\frac{1}{2}\gz{\frac{1+x}{2}+\frac{2x}{1+x}}
\end{equation*}
  is a good choice for $f$. 

On the other hand $f$ is the optimal one because of the 
  construction of $\mu$.
We summarise the results in the following theorem.

\begin{theorem}
For every function $g\in\Fop$
\begin{equation*}
\det (\qCov_{D,f}^s)\ge\det (\qCov_{D,g}^{as})
\end{equation*}
  holds, where
\begin{equation*}
f_{opt}=\frac{1}{2}\gz{\frac{1+x}{2}+\frac{2x}{1+x}}.
\end{equation*}
\end{theorem}

Note that the difference between the left-hand side and the 
  right-hand side of Equation (\ref{eq:globupperboundineq}) is
\begin{equation*}
q(x)=\frac{f_{opt}(0)(1+x)^{2}}{2f_{opt}(x)}-\frac{(1-x)^{2}}{2(1+x)}
  =\frac{8x^{2}}{x^{3}+7x^{2}+7x+1}
\end{equation*}
  and we have $q'(0)=0$ therefore it cannot be estimated below
  by any operator monotone function.


\begin{thebibliography}{10}

\bibitem{And}
A.~Andai.
\newblock {Uncertainty principle with quantum Fisher information}.
\newblock {\em {J. Math. Phys.}}, {49}({1}), {JAN} {2008}.

\bibitem{Aud}
K.~Audenaert, L.~Cai, F.~Hansen.
\newblock{Inequalities for quantum skew information}.
\newblock{\em{Letters in Mathematical Physics}}, {85}({2-3}):135-146, 2008.

\bibitem{BecBel}
E.F. Beckenbach, R.~Bellman, and R.E. Bellman.
\newblock {\em Inequalities}.
\newblock Ergebnisse der Mathematik und ihrer Grenzgebiete. Springer, 1961.

\bibitem{GibHanIso}
P.~Gibilisco, F.~Hansen and T.~Isola.
\newblock On a correspondence between regular and non-regular
operator monotone functions.
\newblock {\em Linear Algebra and Applications},
430(8-9):2225--2232, 2009.

\bibitem{GibFumPetz}
P.~Gibilisco, F.~Hiai, and D.~Petz.
\newblock Quantum covariance, quantum Fisher information
and the uncertainty relations.
\newblock {\em IEEE Trans. Inform. Theory}, 55(1):439--443, 2009.

\bibitem{GibImpIso1}
P.~Gibilisco, D.~Imparato, and T.~Isola.
\newblock Uncertainty principle and quantum fisher information. ii.
\newblock {\em J. Math. Phys.}, 48(7):072108, 2007.

\bibitem{GibImpIso3}
P.~Gibilisco, D.~Imparato, and T.~Isola.
\newblock A Robertson-type uncertainty principle and quantum fisher
  information.
\newblock {\em Linear Algebra Appl.}, 428(7):1706 -- 1724, 2008.

\bibitem{GibIsoNew}
P.~Gibilisco and T.~Isola.
\newblock How to distinguish quantum covariances using uncertainty
relations.
\newblock {\em Journal of Mathematical Analysis and Applications},
384(2):670--676, 2011.

\bibitem{GibIso}
P.~Gibilisco and T.~Isola.
\newblock Uncertainty principle and quantum fisher information.
\newblock {\em Ann. Inst. Statist. Math.}, 59(1):147--159, 2007.

\bibitem{Han}
F.~Hansen.
\newblock Metric adjusted skew information.
\newblock {\em Proc. Natl. Acad. Sci. USA}, 105(29):9909--9916, 2008.

\bibitem{Hei}
W.~Heisenberg.
\newblock {\"{U}ber den anschaulichen Inhalt der quantentheoretischen Kinematik
  und Mechanik}.
\newblock {\em Z. Phys.}, 43(3):172--198, 1927.

\bibitem{HiaPetTot}
F.~Hiai, D.~Petz, and G.~Toth.
\newblock Curvature in the geometry of canonical correlation.
\newblock {\em Studia Sci. Math. Hungar.}, 32(1-2):235--249, 1996.

\bibitem{HorJoh}
R.~A. Horn and C.~R. Johnson.
\newblock {\em Matrix Analysis}.
\newblock Cambridge University Press, 1990.

\bibitem{Ken}
E~.H. Kennard.
\newblock Zur quantenmechanik einfacher bewegungstypen.
\newblock {\em Z. f\"ur Phys.}, 44(4-5):326--352, 1927.

\bibitem{Kos}
H.~Kosaki.
\newblock Matrix trace inequalities related to uncertainty principle.
\newblock {\em Int. J. Math.}, 16(06):629--645, 2005.

\bibitem{Luo1}
S.~Luo.
\newblock Quantum fisher information and uncertainty relations.
\newblock {\em Lett. Math. Phys.}, 53(3):243--251, 2000.

\bibitem{Luo2}
S.~Luo.
\newblock Wigner--Yanase skew information and uncertainty relations.
\newblock {\em Phys. Rev. Lett.}, 91:180403, Oct 2003.

\bibitem{LuoZha2}
S.~Luo and Q.~Zhang.
\newblock On skew information.
\newblock {\em IEEE Trans. Inf. Theory}, 50(8):1778--1782, 2004.

\bibitem{LuoZha3}
S.~Luo and Q.~Zhang.
\newblock Correction to on skew information.
\newblock {\em IEEE Trans. Inf. Theory}, 51(12):4432--4432, 2005.

\bibitem{LuoZha1}
S.~Luo and Z.~Zhang.
\newblock An informational characterization of Schr\"odinger's uncertainty
  relations.
\newblock {\em J. Stat. Phys.}, 114(5-6):1557--1576, 2004.

\bibitem{PetSud}
D.~Petz and Cs. Sud\'ar.
\newblock Geometries of quantum states.
\newblock {\em J. Math. Phys.}, 37(6):2662--2673, 1996.

\bibitem{Rob1}
H.~P. Robertson.
\newblock The uncertainty principle.
\newblock {\em Phys. Rev.}, 34:163--164, Jul 1929.

\bibitem{Rob2}
H.~P. Robertson.
\newblock An indeterminacy relation for several observables and its classical
  interpretation.
\newblock {\em Phys. Rev.}, 46:794--801, Nov 1934.

\bibitem{Sch}
E.~Schr\"odinger.
\newblock About Heisenberg uncertainty relation.
\newblock In {\em Proc.Prussian Acad.Sci.,Phys.Math.
  Section,Vol.XIX,pp.293(1930}.

\bibitem{YanFurKur}
K.~Yanagi, S.~Furuichi, and K.~Kuriyama.
\newblock A generalized skew information and uncertainty relation.
\newblock {\em IEEE Trans. Inf. Theory}, 51(12):4401--4404, 2005.

\end{thebibliography}
\end{document}